\documentclass[12pt]{article}
\usepackage{amsfonts}
\usepackage{amsmath}

\usepackage[english]{babel}
\usepackage[T1]{fontenc}

\usepackage[top=30truemm,bottom=30truemm,left=25truemm,right=25truemm]{geometry}

\newtheorem{theorem}{Theorem}
\newtheorem{definition}{Definition}

\newtheorem{proposition}{Proposition}

\newtheorem{proof}{Proof}

\newtheorem{note}{Note}

\def\<{\langle}
\def\>{\rangle}

\begin{document}

\centerline{\bf On Transmission Efficiency of  Quantum Modulations }

\bigskip\bigskip

\centerline{\sc Kyouhei Ohmura}
\vspace{+1mm}
\centerline{\it Department of Information Sciences,}
\centerline{\it Tokyo University of Science,}
\centerline{\it Noda City, Chiba 278-8510, Japan}
\centerline{E-mail: {\tt 6317701@ed.tus.ac.jp, \quad \tt ohmura.kyouhei@gmail.com}}
\bigskip
\centerline{\sc Noboru Watanabe}
\vspace{+1mm}
\centerline{\it Department of Information Sciences,}
\centerline{\it Tokyo University of Science,}
\centerline{\it Noda City, Chiba 278-8510, Japan}
\centerline{E-mail: {\tt watanabe@is.noda.tus.ac.jp}}

\bigskip\bigskip

\centerline{\bf Abstract }

\bigskip

\noindent In quantum information theory, it is important to find modulations with low information loss for noisy channels. In this paper, using the quantum dynamical entropy and the quantum dynamical mutual entropy, we investigate the transmission efficiency of two quantum modulators through attenuation channels.


\tableofcontents

\section{Introduction}
The dynamical entropy, formulated by Kolmogorov to solve mathematical isomorphic problems, is also important in information theory since it gives the average of information amount of information sources \cite{bili}, \cite{km1}, \cite{km2}, \cite{km3}. This entropy has several noncommutative generalizations, and they are used for classifications of operator algebras \cite{af}, \cite{cnt}, \cite{cs}, \cite{emch}. On the other hand, another important measure in information theory is dynamical mutual entropy. The dynamical mutual entropy gives the average of the amount of information transimitted correctly from the input system to the output system.

In \cite{muraki}, Ohya and Muraki formulated new noncommutative dynamical entropy and dynamical mutual entropy on $C^*$-algebras using the concept of compound states introduced by Ohya \cite{compmutual}. The dynamical entropies are used not only for mathematical classification of operator algebras but also for model calculations of quantum information theory \cite{noteon}.

Incidentally, when considering the transmission of quantum states, it is important to study modulating the initial states into suitable states for the channel. That is, it is important to investigate modulation schemes that efficiently transmit quantum information \cite{itsuse}, \cite{modu}, \cite{makino}.

In this paper, using the quantum dynamical entropy and dynamical mutual entropy given by Ohya and Muraki, when the attenuation channel is used as the channel between input system and output system, we investigate which modulator, PPM (Pulse Position Modulator) or PWM (Pulse Width Modulator), transmits quantum information more efficiency.

The paper is organized as follows: In Sec. 2 we recall the definitions of some quantum channels and the attenuation channel. In Sec. 3 we mention the definitions of basic quantum entropies in order to define the quantum dynamical mutual entropy, and recall some known facts about the entropies. Section 4 is devoted to describe the definitions of quantum dynamical entropy and quantum dynamical mutual entropy given by Ohya and to state their properties. The main result in this paper is Sec. 5, where we discuss the efficiency of the optical modulations (PPM and PWM) with the quantum states by using the entropy ratio given by the quantum dynamical entropies.


\section{Quantum Channels}
In this section, we briefly recall the notions of several quantum channels.\\

\noindent Let $(\mathcal{A}, \mathfrak{S}(\mathcal{A}), \alpha (G))$ be an input quantum dynamical system and $(\mathcal{B}, \mathfrak{S} (\mathcal{B}), \beta (G' ))$, be an that of output. Namely,  $\mathcal{A}$ (resp. $\mathcal{B}$) is a $C^{*}$-algebra, $\mathfrak{S}(\mathcal{A})$ (resp. $\mathfrak{S} (\mathcal{B})$) is the set of all states on $\mathcal{A}$ (resp. $\mathcal{B}$) and $\alpha(G)$ (resp. $\beta (G')$) is the set of all *-automorphisms on $\mathcal{A}$ (resp. $\mathcal{B}$) associate with the group $G$ (resp. $G'$). The above triplet represents the dynamics of the quantum system .  \\

\begin{definition}
A map $\Lambda^{*}$ from $\mathfrak{S}(\mathcal{A})$ to $\mathfrak{S} (\mathcal{B})$ is called a {\it channel}.
\end{definition}

\begin{definition}
If
\begin{equation}
\Lambda^{*}(\sum_n \lambda_n \varphi_n) = \sum_n \lambda_n \Lambda^{*} (\varphi_n)
\end{equation}
holds for any $\varphi_n \in \mathfrak{S}(\mathcal{A})$, and $\sum_n \lambda_n = 1, \lambda_n \ge 0 \ (\forall n \in \mathbb{N})$, $\Lambda^{*}$ is called a {\it linear channel}.
\end{definition}

\begin{definition}
$\Lambda^{*}$ denotes the dual map of $\Lambda : \mathcal{B} \to \mathcal{A}$, i.e.
$$
\Lambda^{*}(\varphi)(B) = \varphi(\Lambda (B))
$$
for any $\varphi \in \mathfrak{S}(\mathcal{A})$ and any $B \in \mathcal{B}$. If $\Lambda$ satisfies
\begin{equation}
\sum_{i,j=1}^{n} A_{i}^{*}\Lambda (B_{i}^{*}B_j )A_j \ge 0
\end{equation}
for any $n \in \mathbb{N}$, any $B_j \in \mathcal{B}$, and any $A_j \in \mathcal{A}$, 
$\Lambda^{*}$ is called a {\it completely positive channel} ({\it c.p. channel} for short).
\end{definition}
It is known that the c.p. channels $\Lambda^* : \mathfrak{S}(\mathcal{A}) \to \mathfrak{S}(\mathcal{B})$ can describe the physical transformations of several quantum systems \cite{itsuse}, \cite{qent}, \cite{makino}.


\subsection{Attenuation Channel}
In communication process, we have to consider the loss of information in the course of information transmission. The attenuation channel given by Ohya and Watanabe \cite{modu} is a mathematical representation of a quantum channel whose noise is given by vacuum state.\\

\noindent Let $\mathcal{A} = \mathbf{B}(\mathcal{H}_1)$ (resp. $\bar{\mathcal{A}} = \mathbf{B}(\mathcal{H}_2))$. Therefore $ \mathfrak{S}(\mathcal{A}) = \mathfrak{S}(\mathcal{H}_1)$ (resp. $\mathfrak{S}(\bar{\mathcal{A}}) = \mathfrak{S}(\mathcal{H}_2))$ is the set of all density operators on a Hilbert space $\mathcal{H}_1 $(resp. $\mathcal{H}_2)$. Furthermore, let $\mathcal{B}$ (resp. $\bar{\mathcal{B}}) $ be a $C^{*}$-algebra on another Hilbert space $\mathcal{K}_1$(resp. $ \mathcal{K}_2) $. Then each state spaces correspond to each physical systems as follows:
\begin{enumerate}
\item $\mathfrak{S}(\mathcal{H}_1)$ : Input system.
\item $\mathfrak{S}(\mathcal{K}_1)$ : Noisy system.
\item $\mathfrak{S}(\mathcal{K}_2)$ : Loss system.
\item $\mathfrak{S}(\mathcal{H}_2)$ : Output system.
\end{enumerate}

\noindent Now $| n\>$ ($n \in \mathbb{Z}_+$) denotes a $n$-th number photon vector state in $\mathcal{H}_i$ or $\mathcal{K}_i$ ($i = 1, 2$) and $V_0$ denotes a mapping from $\mathcal{H}_1 \otimes \mathcal{K}_1$ to $ \mathcal{H}_2 \otimes \mathcal{K}_2$:
\begin{eqnarray}
V_0 (|n \> \otimes |0\rangle ) := \sum_{j}^{n}C_{j}^{n}|j\rangle  \otimes |n - j\rangle ,
\end{eqnarray}
where
\begin{eqnarray}
C_{j}^{n} := \sqrt{ \frac{n !}{j! (n -j)!}} \alpha^j (-\beta)^{n -j}\quad , \quad \alpha^2 + \beta^2 = 1.
\end{eqnarray}
Using $V_0$, one obtain the CP channel $\pi_0^* : \mathfrak{S}(\mathcal{H}_1 \otimes \mathcal{K}_1) \to \mathfrak{S}(\mathcal{H}_2 \otimes \mathcal{K}_2 )$ given by
$$
\pi_0^* (\cdot) := V_0 (\cdot) V_0^* .
$$

\begin{definition}
Under the above settings, {\it attenuation channel} $\Lambda_{0}^{*} : \mathfrak{S}(\mathcal{H}_1) \to \mathfrak{S}(\mathcal{H}_2)$ {\it with a vacuum state} $\zeta_0 := |0\>\< 0|$ was defined by
\begin{equation}\label{atte}
\Lambda_0^{*}(\rho) := \mathrm{Tr}_{\mathcal{K}_{2}}\pi_0^{*}(\rho \otimes \zeta_0) = \mathrm{Tr}_{\mathcal{K}_2} V_0 (\rho \otimes |0 \rangle \langle 0|) V_0^{*} .
\end{equation}
\end{definition}

\noindent Then 
\begin{equation}\label{eta}
\eta = |\alpha |^{2}  \qquad (0 \le \eta \le 1)
\end{equation}
is called a {\it transition ratio} of the attenuation channel $\Lambda_{0}^{*}.$


\section{Quantum Entropies}
We introduce some definitions and related theorems of entropies needed for formulations of the quantum dynamical entropies of next section.

\subsection{$\mathcal{S}$-Mixing Entropy}
\noindent In \cite{smix}, Ohya generalized von Neumann entropy to $C^*$-algebras.\\
\noindent Let $(\mathcal{A}, \mathfrak{S}(\mathcal{A}), \alpha(G))$ be a $C^{*}$-dynamical system and $\mathcal{S}$ be a weak* compact and convex subset of $\mathfrak{S}(\mathcal{A})$.
\begin{note}
$\mathfrak{S}(\mathcal{A})$, $I(\alpha)$ (the set of all invariant states for $\alpha$) and $K(\alpha)$ (the set of all KMS states) are  weak* compact and convex subset of $\mathfrak{S}(\mathcal{A})$.
\end{note}

\noindent Let ${\rm ex} \mathcal{S}$ be the set of all extreme points of $\mathcal{S}$. From the Krein-Mil'man theorem \cite{krein}, there holds ${\rm ex}\mathcal{S} \neq \phi$. Every state $\varphi \in \mathcal{S}$ has a maximal measure $\mu$ pseudosupported on $\mathrm{ex}\mathcal{S}$ such that
\begin{equation}\label{bc}
\varphi = \int_{\mathrm{ex} \mathcal{S}} \omega d\mu .
\end{equation}
The measure $\mu$ giving the above decomposition is not unique unless $\mathcal{S}$ is a Choquet simplex. Then $M_{\varphi}(\mathcal{S})$ denotes the set of all such measures. Moreover, if  $\mu \in M_{\varphi} (\mathcal{S})$ has countable supports, that is, there holds
$$
\mu = \sum \lambda_k \delta_{\varphi_k},
$$
where $\lambda_k > 0$, $\sum \lambda_k = 1$, and $\{ \varphi_k \} \subset {\rm ex} \mathcal{S}$, we put the set of  all such measures by $D_{\varphi} (\mathcal{S})$.

\begin{definition}\label{Smix}
Under the above settings, the entropy of $\varphi \in \mathcal{S}$ is given by
\begin{equation}
S^{\mathcal{S}} (\varphi ) := 
\begin{cases}
\inf \{ - \sum \lambda_k \log \lambda_k ;\ \mu = \{ \lambda_k \} \in D_{\varphi} (\mathcal{S}) \} \\
+ \infty \qquad (\mu \notin D_{\varphi} (\mathcal{S}))
\end{cases}
\end{equation}
\end{definition}

\noindent This entropy is called {\it $\mathcal{S}$-mixing entropy} and describes the amount of information of the state $\varphi$ measured from the subsystem $\mathcal{S}$. Then the following theorem holds.

\begin{theorem}
If $\mathcal{A} = \mathbf{B} (\mathcal{H})$ and $\mathcal{S} = \mathfrak{S}(\mathcal{H})$, $\mathcal{S}$-mixing entropy corresponds to von Neumann entropy \cite{vn}, i.e.
\begin{equation}
S^{\mathcal{S}}(\varphi) = - {\rm Tr}\rho \log \rho \quad ,\quad  \rho \in \mathfrak{S}(\mathcal{H}).
\end{equation}
\end{theorem}

\noindent By taking the set of all quantum channels as $\mathcal{S}$, Mukhamedov and Watanabe defined a general extension of the $\mathcal{S}$-mixing entropy and obtained important results for entangled states \cite{mw}.

\subsection{Relative Entropy of States}

In information theory, the relative entropy is an information which represents the complexity of a state with respect to another state. In \cite{condexp}, Umegaki introduced the relative entropy (which is so-called {\it quantum relative entropy}) for $\sigma$-infinite and semifinite von Neumann algebras.

\begin{definition}
For two density operators $\rho$ and $\sigma$, {\it Umegaki relative entropy} is defined as
\begin{equation}\label{umerela}
S(\rho , \sigma) = 
\begin{cases}
\mathrm{Tr} \rho (\log \rho - \log \sigma) & {\rm if } \ \mathrm{supp}\sigma \geq \mathrm{supp}\rho \\
+\infty & {\rm otherwise}.
\end{cases}
\end{equation}
\end{definition}

\noindent Araki generalized the relative entropy (\ref{umerela}) to the general von Neumann algebras using the relative modular operator \cite{araki}.  Moreover, the relative entropy on *-algebras was formulated by Uhlmann \cite{uhl}.

\subsection{Mutual Entropy of States}
The notion of mutual entropy is the amount of information correctly transmitted from the input system to the output . In \cite{compmutual}, the quantum analogue of the mutual entropy was defined by Ohya with respect to density operators. Furthermore, he generalized the notion of quantum mutual entropy for $C^*$-dynamical systems.\\

\noindent Let $(\mathcal{A}, \mathfrak{S}(\mathcal{A}), \alpha (G))$ and $(\mathcal{B}, \mathfrak{S}(\mathcal{B}), \beta (G'))$ be unital $C^{*}$-dynamical systems (i.e. with the identity), and $\mathcal{S}$ be a weak* compact convex subset of $\mathfrak{S}(\mathcal{A}).$
\begin{definition}
For an initial state $\varphi \in \mathcal{S}$ and a channel $\Lambda^{*} : \mathfrak{S}(\mathcal{A}) \to \mathfrak{S}(\mathcal{B})$, two {\it compound states} are given by
\begin{equation}\label{true}
\Phi_{\mu}^{\mathcal{S}} := \int_{\mathcal{S}} \omega \otimes \Lambda^{*}\omega d\mu ,
\end{equation}
\begin{equation}
\Phi_0 := \varphi \otimes \Lambda^{*} \varphi.
\end{equation}
\end{definition}
The compound state $\Phi_{\mu}^{\mathcal{S}}$ represents the correlation between the input state $\varphi$ and the output state $\Lambda^{*}\varphi$. On the other hand, one can see that $\Phi_0$ doesn't express the correlation. \\
\noindent Then the mutual entropy with respect to $\mathcal{S}$ and $\mu$ is given by
$$
I_{\mu}^{\mathcal{S}}(\varphi , \Lambda^{*}) = S(\Phi_{\mu}^{\mathcal{S}}, \Phi_0),
$$
where $S(\cdot , \cdot)$ is the Araki's relative entropy.

\begin{definition}
Under the above notations, the mutual entropy with respect to $\mathcal{S}$ is given by
\begin{equation}
I^{\mathcal{S}}(\varphi \ ; \Lambda^{*}) = \sup \{ I_{\mu}^{\mathcal{S}}(\varphi , \Lambda^{*}) \ ; \ \mu \in M_{\varphi}(\mathcal{S}) \}.
\end{equation}
\end{definition}
\noindent When $\mathcal{S}$ is the total space $\mathfrak{S}(\mathcal{A})$, we simpley denote $I (\varphi \ ; \Lambda^*)$ and $S(\varphi)$.\\

\noindent Now we show the definition of mutual entropy if the state defined by a density operator.\\
Let $\mathcal{A} = \mathbf{B} (\mathcal{H})$. Then any normal state $\varphi$ can be written as  $\varphi (A) = \mathrm{Tr} \rho A$ ($\forall A \in \mathcal{A}$) using the corresponding the density operator $\rho$.
 Every Schatten decomposition \cite{sch} $\rho = \sum_n \lambda_n E_n  ,\ E_n = |x_n \rangle \langle x_n |$ provides every orthogonal measures in $D_{\varphi}(\mathfrak{S}(\mathcal{A}))$. Since the Schatten decomposition of $\rho$ is not unique unless every eigenvalue $\lambda_n$ is nondegenerate, the compound state $\Phi_{\mu}^{\mathcal{S}}$ (\ref{true}) is expressed as 
$$
\Phi_E (\mathcal{Q}) := \Phi_{\mu}^{\mathcal{S}}(\mathcal{Q}) = \mathrm{Tr} \sigma_E \mathcal{Q} \quad ,\quad \mathcal{Q} \in \mathcal{A} \otimes \mathcal{B}
$$
with
$$
\sigma_E := \sum_n \lambda_n E_n \otimes \Lambda^{*}E_n,
$$
where $E$ represents the Schatten decomposition $\{ E_n \}$. 

\begin{definition}
Then the mutual entropy for $\rho$ and the channel $\Lambda^{*}$ is given by
\begin{equation}
I (\rho \ ; \Lambda^{*}) = \sup \{ S(\sigma_E , \sigma_0) \ ;\ E = \{ E_n \} \  {\rm of} \ \rho \},
\end{equation}
where $S(\cdot , \cdot)$ is the Umegaki relative entropy and $\sigma_0 := \rho \otimes \Lambda^{*}\rho$.
\end{definition}

\noindent For $I (\varphi \ ; \Lambda^{*})$, Ohya proved the following theorem called the {\it fundamental inequalities} \cite{compmutual}. 

\begin{theorem}
\begin{equation}\label{fund}
0 \le I(\varphi \ ; \Lambda^{*}) \le \min \{ S(\varphi),\ S(\Lambda^{*}\varphi) \}.
\end{equation}
\end{theorem}
\noindent This theorem implies that the amount of information correctly transmitted does not exeed the amount of information of the input and that of the output.


\section{Quantum Dynamical Mutual Entropy}
In this section, we briefly review some notions concerning the quantum dynamical entropy and quantum dynamical mutual entropy. These results are described in \cite{muraki}, \cite{mop}, \cite{some}, \cite{state}.\\
 
\noindent A stationary quantum information source is described by the triplet $(\mathcal{A}, \mathfrak{S}(\mathcal{A}), \theta_{\mathcal{A}})$ and a stationary state $\varphi$ with respect to the *-automorphism $\theta_{\mathcal{A}}$  on $\mathcal{A}$ (i.e. $\varphi \circ \theta_{\mathcal{A}} = \varphi$). Let  $(\mathcal{B}, \mathfrak{S}(\mathcal{B}), \theta_{\mathcal{B}})$ be an output $C^{*}$-dynamical system and  $\Lambda^{*} : \mathfrak{S}(\mathcal{A}) \to \mathfrak{S}(\mathcal{B})$ be a covariant channel which is a dual of a completely positive unital map (c.p.u. map for short) $\Lambda : \mathcal{B} \to \mathcal{A}$ such that $\Lambda \circ \theta_{\mathcal{B}} = \theta_{\mathcal{A}} \circ \Lambda$.\\
Now we construct compound states (\ref{true}) on the two dynamical systems. Let $\alpha^{M} = (\alpha_{1}, \alpha_{2}, \cdots , \alpha_{M})$ and $\beta^{N} = (\beta_{1}, \beta_{2}, \cdots , \beta_{N})$ be finite sequences of c.p.u. maps
$$
\alpha_{m} : \mathcal{A}_{m} \to \mathcal{A} \qquad,\qquad \beta_n : \mathcal{B}_n \to \mathcal{B}
$$
where $\mathcal{A}_m$ and $\mathcal{B}_n$ $(m = 1, \cdots , M,\ n = 1, \cdots , N)$ are finite dimensional unital $C^{*}$-algebras. Let $\mathcal{S}$ be a weak * convex subset of $\mathfrak{S}(\mathcal{A})$ and $\varphi$ be a state in $\mathcal{S}$. For $\alpha^M$ and an extremal decomposition measure $\mu$ of $\varphi$, we obtain the compound state of $\alpha_{1}^{*}\varphi , \alpha_{2}^{*}\varphi , \cdots , \alpha_{M}^{*} \varphi$  on the tensor product algebra $\bigotimes_{m=1}^{M}\mathcal{A}_m$ as
$$
\Phi_{\mu}^{\mathcal{S}}(\alpha^{M}) := \int_{\mathcal{S}} \bigotimes_{m=1}^{M} \alpha_{m}^{*}\omega d\mu (\omega).
$$
$\alpha : \mathcal{A}_0 \to \mathcal{A} $ (resp. $\beta : \mathcal{B}_0 \to \mathcal{B})$ denotes c.p.u. map from a finite dimensional unital $C^{*}$-algebra $\mathcal{A}_0$ (resp. $\mathcal{B}_0$) to $\mathcal{A}$ (resp. $\mathcal{B}$). Define
$$
\alpha^N := (\alpha , \theta_{\mathcal{A}}\circ \alpha , \cdots , \theta_{\mathcal{A}}^{N-1} \circ \alpha),
$$
$$
\beta_{\Lambda}^N := (\Lambda \circ \beta , \Lambda \circ \theta_{\mathcal{B}} \circ \beta , \cdots , \Lambda \circ \theta_{B}^{N-1} \circ \beta ).
$$
Similarly we have the compound states which represents the correlation of the states on the output $\bigotimes_{n=1}^N \mathcal{B}_n$:
$$
\Phi_{\mu}^{\mathcal{S}}(\beta^{N}) := \int_{\mathcal{S}}\bigotimes_{n=1}^{N} \beta_{n}^{*}\omega d\mu (\omega).
$$

\noindent Furthermore $\Phi_{\mu}^{\mathcal{S}}(\alpha^{M} \cup \beta^{N}) $ is a compound state of $\Phi_{\mu}^{\mathcal{S}}(\alpha^M)$ and $\Phi_{\mu}^{\mathcal{S}}(\beta^{N})$ with $ \alpha^M \cup \beta^{N} \equiv (\alpha_1 , \alpha_2 , \cdots , \alpha_M , \beta_1 , \beta_2 , \cdots , \beta_N)$ constructed as
$$
\Phi_{\mu}^{\mathcal{S}}(\alpha^{M} \cup \beta^{N}) := \int_{\mathcal{S}}(\bigotimes_{m=1}^{M}\alpha_{m}^{*} \omega) \otimes (\bigotimes_{n=1}^{N} \beta_{n}^{*} \omega) d\mu.
$$

\begin{definition}
For any pair $(\alpha^{M}, \beta^{N})$ and any extremal decomposition measure $\mu$ of $\varphi$, the entropy functional $S_{\mu}$ and the mutual entropy functional $I_{\mu}$ are defined by
$$
S_{\mu}^{\mathcal{S}}(\varphi \ ; \alpha^{M}) := \int_{\mathcal{S}} S(\bigotimes_{m=1}^{M} \alpha_{m}^{*}\omega, \Phi_{\mu}^{S}(\alpha^{M}))d\mu (\omega),
$$
$$
I_{\mu}^{\mathcal{S}}(\varphi \ ; \alpha^{M}, \beta^{N}) := S(\Phi_{\mu}^{\mathcal{S}}(\alpha^{M} \cup \beta^{N}), \Phi_{\mu}^{\mathcal{S}}(\alpha^{M})\otimes \Phi_{\mu}^{\mathcal{S}}(\beta^{N})),
$$
respectively, where $S(\cdot ,\cdot)$ is the Araki's relative entropy.
\end{definition}

\noindent Moreover, the functional $S^{\mathcal{S}}(\varphi \ ; \alpha^M)$ (resp. $I^{\mathcal{S}}(\varphi \ ; \alpha^{M}, \beta^{N}))$ is given by taking the supremum of $S_{\mu}^{\mathcal{S}}(\varphi \ ; \alpha^{M})$ (resp. $I_{\mu}^{\mathcal{S}}(\varphi \ ; \alpha^{M}, \beta^{N})$) for all possible extremal decompositions $\mu$ of $\varphi$ :
$$
S^{\mathcal{S}}(\varphi \ ; \alpha^M) := \sup \{ S_{\mu}^{\mathcal{S}}(\varphi \ ; \alpha^{M})\ ;\ \mu \in M_{\varphi} (\mathcal{S})\},
$$
$$
I^{\mathcal{S}}(\varphi \ ; \alpha^{M}, \beta^{N}) := \sup \{  I_{\mu}^{\mathcal{S}}(\varphi \ ; \alpha^{M}, \beta^{N})\ ;\ \mu \in M_{\varphi}(\mathcal{S}) \} .
$$

\noindent Under the above notations, $\tilde{S}^{\mathcal{S}}(\varphi \ ; \theta_{\mathcal{A}}, \alpha)$ and $\tilde{I}^{\mathcal{S}}(\varphi \ ; \Lambda^{*}, \theta_{\mathcal{A}}, \theta_{\mathcal{B}}, \alpha , \beta)$ are given by
$$
\tilde{S}^{\mathcal{S}}(\varphi \ ; \theta_{\mathcal{A}}, \alpha) :=  \lim_{N \to \infty} \frac{1}{N}S^{\mathcal{S}}(\varphi \ ; \alpha^{N}),
$$
$$
\tilde{I}^{\mathcal{S}}(\varphi \ ; \Lambda^{*}, \theta_{\mathcal{A}}, \theta_{\mathcal{B}}, \alpha , \beta) := \lim_{N \to \infty} \frac{1}{N} I^{\mathcal{S}}(\varphi \ ; \alpha^{M}, \beta^{N}).
$$

\begin{definition}
The {\it quantum dynamical entropy} $\tilde{S}^{\mathcal{S}}(\varphi \ ; \theta_{\mathcal{A}})$ and the {\it quantum dynamical mutual entropy} $\tilde{I}^{\mathcal{S}}(\varphi \ ; \Lambda^{*}, \theta_{\mathcal{A}}, \theta_{\mathcal{B}})$ are defined by taking the supremum for all possible $\mathcal{A}_0$'s, $\alpha$'s, $\mathcal{B}_0$'s, and $\beta$'s : 
\begin{eqnarray}
\tilde{S}^{\mathcal{S}}(\varphi\  ; \theta_{\mathcal{A}}) &:=& \sup_{\alpha} \tilde{S}^{\mathcal{S}}(\varphi \ ; \theta_{\mathcal{A}}, \alpha), \\
\tilde{I}^{\mathcal{S}}(\varphi \ ; \Lambda^{*}, \theta_{\mathcal{A}}, \theta_{\mathcal{B}}) &:=& \sup_{\alpha , \beta}\tilde{I}^{\mathcal{S}}(\varphi \ ; \Lambda^{*}, \theta_{\mathcal{A}}, \theta_{\mathcal{B}}, \alpha , \beta).
\end{eqnarray}
\end{definition}
\noindent Then the fundamental inequalities (\ref{fund}) holds for $\tilde{S}^{\mathcal{S}}(\varphi \ ; \theta_{\mathcal{A}})$ and $\tilde{I}^{\mathcal{S}}(\varphi \ ; \Lambda^{*}, \theta_{\mathcal{A}}, \theta_{\mathcal{B}})$. .

\begin{proposition}
\begin{equation}\label{fund2}
0 \le \tilde{I}^{\mathcal{S}}(\varphi \ ; \Lambda^{*}, \theta_{\mathcal{A}} ,\theta_{\mathcal{B}}) \le \min \{ \tilde{S}^{\mathcal{S}}(\varphi \ ; \theta_{\mathcal{A}}), \tilde{S}^{\mathcal{S}}(\Lambda^{*}\varphi \ ; \theta_{\mathcal{B}}) \} .
\end{equation}
\end{proposition}

\noindent Furthermore, it is known that $\tilde{S}^{\mathcal{S}}(\varphi \ ; \theta_{\mathcal{A}})$ and $\tilde{I}^{\mathcal{S}}(\varphi \ ; \Lambda^{*}, \theta_{\mathcal{A}} ,\theta_{\mathcal{B}})$ include the {\it Kolmogorov entropy} (or {\it Kolmogorov-Sinai entropy}) as the special case.

\begin{proposition}
If $\mathcal{A}_k, \mathcal{A}$ are commutative $C^{*}$-algebras and each $\alpha_k$ is an embedding, then our functionals coincide with the classical cases: 
\begin{eqnarray}
S_{\mu}^{\mathcal{S}(\mathcal{A})}(\varphi \ ; \alpha^{M}) &=& S_{\mu}^{classical}(\bigvee_{m=1}^{M} \tilde{A}_m), \\
I_{\mu}^{\mathcal{S}(\mathcal{A})}(\varphi \ ; \alpha^{M}, \beta_{id}^{N}) &=& I_{\mu}^{classical}(\bigvee_{m=1}^{M}\tilde{A}_m , \bigvee_{n=1}^{N}\tilde{B}_n)
\end{eqnarray}
for any finite partitions $\tilde{A}_m , \tilde{B}_n$ of a probability space $(\Omega , \mathcal{F}, \varphi).$
\end{proposition}

\noindent Moreover, the following Kolmogorov-Sinai type convergence theorems hold.

\begin{theorem}
Let $\alpha_m$ be a sequence of c.p. maps $\alpha_m : \mathcal{A}_m \to \mathcal{A} $ and $\beta_m : \mathcal{B}_m \to \mathcal{B} $ such that there exist c.p. maps $\alpha'_m : \mathcal{A} \to \mathcal{A}_m$ satisfying $\alpha_m \circ \alpha'_m \to id_{\mathcal{A}}$ in the pointwise topology. Then there holds:
\begin{equation}
\tilde{S}^{\mathcal{S}}(\varphi \ ; \theta_{\mathcal{A}}) = \lim_{m \to \infty}\tilde{S}^{\mathcal{S}}(\varphi \ ; \theta_{\mathcal{A}}, \alpha_m).
\end{equation}
\end{theorem}

\begin{theorem}
Let $\alpha_m$ and $\beta_m$ be sequences of c.p. maps $\alpha_m : \mathcal{A}_m \to \mathcal{A}$ and $\beta_m : \mathcal{B}_m \to \mathcal{B}$ such that there exist c.p. maps $\alpha'_m : \mathcal{A} \to \mathcal{A}_m$ and $\beta'_m : \mathcal{B} \to \mathcal{B}_m$ satisfying $\alpha_m \circ \alpha'_m \to id_{\mathcal{A}}$ and $\beta_m \circ \beta'_m \to id_{\mathcal{B}}$ in the pointwise topology. Then one obtain
\begin{equation}
\tilde{I}^{\mathcal{S}}(\varphi \ ; \Lambda^{*}, \theta_{\mathcal{A}}, \theta_{\mathcal{B}}) = \lim_{m \to \infty} \tilde{I}^{\mathcal{S}}(\varphi \ ; \Lambda^{*}, \theta_{\mathcal{A}}, \theta_{\mathcal{B}}, \alpha_m , \beta_m).
\end{equation}
\end{theorem}


\subsection{Quantum Dynamical Mutual Entropy for Density Operators}
Based on the above construction, we rewrite the dynamical entropies in terms of density operators.\\

\noindent Let $\mathbf{B}(\mathcal{H}_0)$ (resp. $\mathbf{B} (\bar{\mathcal{H}_0})$) be the set of all bounded linear operators on a Hilbert space $\mathcal{H}_0$ (resp. $\bar{\mathcal{H}_0}$) and $\mathcal{A}_0$ (resp. $\mathcal{B}_0$) be a finite subset in $\mathbf{B} (\mathcal{H}_0)$ (resp. $\mathbf{B}(\bar{\mathcal{H}_0})$). Furthermore, let $\mathcal{A}$ (resp. $\mathcal{B}$) be an infinite tensor product space of $\mathbf{B} (\mathcal{H}_0)$ (resp. $\mathbf{B}(\bar{\mathcal{H}_0})$) represented by
$$
\mathcal{A} := \otimes_{i = -\infty}^{\infty} \mathbf{B} (\mathcal{H}_0),
$$
$$
\mathcal{B} := \otimes_{i = -\infty}^{\infty} \mathbf{B} (\bar{\mathcal{H}_0}).
$$

\noindent Moreover,  we define a shift transformation on $\mathcal{A}$ (resp. $\mathcal{B}$) by $\theta_{\mathcal{A}}$ (resp. $\theta_{\mathcal{B}})$, that is,
$$
\theta_{\mathcal{A}}(\otimes_{i=-\infty}^{\infty} A_{i}) := \otimes_{i' = -\infty}^{\infty}A_{i'} \quad (i' = i - 1)\ , \quad \forall \ \otimes_{i = -\infty}^{\infty}A_i \in \mathcal{A},
$$
$$
\theta_{\mathcal{B}}(\otimes_{j=-\infty}^{\infty} B_{j}) := \otimes_{j' = -\infty}^{\infty}B_{j'} \quad (j' = j - 1) \ , \quad \forall  \ \otimes_{j = -\infty}^{\infty}B_j \in \mathcal{B}.
$$

\noindent $\alpha$ (resp. $\beta$) denotes the embedding from $\mathcal{A}_0$ to $\mathcal{A}$, (resp. $\mathcal{B}_0$ to $\mathcal{B}$):
$$
\alpha(A) := \cdots I \otimes I \otimes A \otimes I \otimes \cdots \in \mathcal{A} \quad ,\quad \forall A \in \mathcal{A}_0 ,
$$
$$
\beta(B) := \cdots I \otimes I \otimes B \otimes I \otimes \cdots \in \mathcal{B} \quad ,\quad \forall B \in \mathcal{B}_0 .
$$

\noindent Let by $\mathfrak{S}_0$ (resp. $\bar{\mathfrak{S}_0})$ be the set of all density operators on $\mathcal{H}_0$ (resp. $\bar{\mathcal{H}_0})$ and $\mathfrak{S}$ (resp. $\bar{\mathfrak{S}}$) be the set of all states $\rho \in \otimes_{i=-\infty}^\infty \mathfrak{S}_0$ (resp. $\bar{\rho}\in \otimes_{i=-\infty}^\infty \bar{\mathfrak{S}_0} )$.\\

\noindent Under the above notations, the dual maps $\theta_{\mathcal{A}}^*, \theta_{\mathcal{B}}^* , \alpha^*, \beta^* $ of $\theta_{\mathcal{A}}, \theta_{\mathcal{B}}, \alpha , \beta$ are obtained as follows:
\begin{enumerate}
\item $\theta_{\mathcal{A}}^{*}$ is a map from $\mathfrak{S}$ to $\mathfrak{S}$ satisfying
$$
\theta_{\mathcal{A}}^* (\otimes_{i=-\infty}^\infty \rho_i) = \otimes_{i' = -\infty}^\infty \rho_{i'} \quad (i' = i+1) \quad , \quad \forall \otimes_{i=-\infty}^\infty \rho_i \in \mathfrak{S}.
$$
\item $\theta_{\mathcal{B}}^{*}$ is a map from $ \bar{\mathfrak{S}}$ to $\bar{\mathfrak{S}}$ satisfying
$$
\theta_{\mathcal{B}}^* (\otimes_{j=-\infty}^\infty \bar{\rho_j}) = \otimes_{j' = -\infty}^\infty \bar{\rho}_{j'} \quad (j'=j+1) \quad ,\quad \forall \otimes_{j=-\infty}^\infty \bar{\rho_j}\in \bar{\mathfrak{S}}.
$$
\item $\alpha^{*}$ is a map from $\mathfrak{S}$ to $\mathfrak{S}_0$ such as
$$
\alpha^* (\otimes_{i=-\infty}^\infty\rho_i) = \mathrm{Tr}_{i\not = 0}(\otimes_{i= -\infty}^\infty \rho_i) = \rho_0 \quad , \quad \forall \otimes_{i=-\infty}^\infty \rho_i \in \mathfrak{S}.
$$
\item $\beta^{*}$ is a map from $\bar{\mathfrak{S}}$ to $\bar{\mathfrak{S}_0}$ such as
$$
\beta^* (\otimes_{j=-\infty}^\infty \bar{\rho_j}) = \mathrm{Tr}_{j\not = 0} (\otimes_{j=-\infty}^{\infty} \bar{\rho_j})=\bar{\rho_0} \quad , \quad \forall \otimes_{j=-\infty}^\infty \bar{\rho_j} \in \bar{\mathfrak{S}}.
$$
where $\mathrm{Tr}_{i\not = 0}$ means to take a partial trace except $i = 0$.
\end{enumerate}

\noindent Now we rewrite the quantum dynamical mutual entropy in density operators case as follows: \\
Put
$$
\alpha^N := (\alpha , \theta_{\mathcal{A}}\circ \alpha , \cdots , \theta_{\mathcal{A}}^{N-1} \circ \alpha),
$$
$$
\beta_{\Lambda}^N := (\Lambda \circ \beta , \Lambda \circ \theta_{\mathcal{B}} \circ \beta , \cdots , \Lambda \circ \theta_{B}^{N-1} \circ \beta ),
$$
where $\Lambda^* = \otimes_{i=-\infty}^{\infty} \Lambda^*$ is a channel from $\mathfrak{S}$ to  $\bar{\mathfrak{S}}$. For any $\otimes_{i=-\infty}^{\infty} \rho_i \in \mathfrak{S}$, an input compound state $\Phi_{\mu}^{\mathcal{S}}(\rho; \alpha^{N})$ with respect to $\alpha^{*}(\rho), \cdots , \alpha^{*} \circ \theta_{\mathcal{A}}^{*N-1}(\rho)$ is defined as
\begin{equation}\label{comp}
\Phi_E (\rho ; \alpha^N) := \sum_{n=1}^{M}\lambda_n \otimes_{i=0}^{N-1}\alpha^{*} \circ \theta_{\mathcal{A}}^{*i}(\otimes_{i=-\infty}^{\infty}E_{n_i}) = \otimes_{i=0}^{N-1} \rho_i .
\end{equation}
When a Schatten decomposition of $\rho_i \in \mathfrak{S}_0 \ (i = 0, ... , N - 1)$ is given by
$$
\rho_i = \sum_{n_i = 1}^{M_i}\lambda_{n_i}E_{n_i} \quad ,\quad (\sum_{n_i = 1}^{M_i}\lambda_{n_i} = 1,\ 0 \le \lambda_{n_i} \le 1), 
$$
the compound state (\ref{comp}) is expressed as
$$
\Phi_E (\rho ; \alpha^N) = \sum_{n_0 = 1}^{M} \cdots \sum_{n_{N-1}= 1}^{M} (\prod_{k=0}^{N-1}\lambda_{n_k}) (\otimes_{i=0}^{N-1}E_{n_i}) .
$$
For an initial state $\rho \in \mathfrak{S}$, we have an output compound state $\Phi_E (\rho ; \beta_{\Lambda}^N)$  with respect to $\beta^* \circ \Lambda^* (\rho), ... , \beta^* \circ \theta_{\mathcal{B}}^{*N-1} \circ \Lambda^* (\rho)$ as
$$
\Phi_E (\rho ; \beta_{\Lambda}^N) := \otimes_{i=0}^{N-1} \beta^* \circ \theta_{\mathcal{B}}^{*i}\Lambda^* (\rho) = \otimes_{i=0}^{N-1} \Lambda^* \rho_i .
$$
\begin{definition}
For any state $\rho = \otimes_{i=-\infty}^{\infty} \rho_i \in \mathfrak{S}$, the correlated compound state with respect to $\Phi_E (\rho \ ; \alpha^N) $ and $\Phi_E (\rho \ ; \beta_{\Lambda}^N)$ is given by
\begin{equation}
\Phi_E (\rho \ ; \alpha^N)\otimes \Phi_E (\rho \ ; \beta_{\Lambda}^N) = (\otimes_{i=0}^{N-1} \rho_i)\otimes(\otimes_{i=0}^{N-1}\Lambda^* \rho_i).
\end{equation}
The state $\Phi_E (\rho \ ; \alpha^N \cup\Lambda \beta^N)$ which represents correlation between two dynamical systems is written as
\begin{equation}
\Phi_E (\rho \ ; \alpha^N \cup\Lambda \beta^N) := \sum_{n_0 =1}^M ... \sum_{n_{N-1}=1}^M (\prod_{k=0}^{N-1} \lambda_{n_k})(\otimes_{i=0}^{N-1} E_{n_i})\otimes(\otimes_{i' = 0}^{N-1}\Lambda^* E_{n_{i'}}).
\end{equation}
\end{definition}

\begin{definition}
\noindent For any initial state $\rho = \otimes_{i=-\infty}^{\infty} \rho_i \in \mathfrak{S}$, the functionals $I_E (\rho \ ; \alpha^N , \beta_{\Lambda}^N)$, $I(\rho \ ; \alpha^N , \beta_{\Lambda}^N)$, $S_E (\rho \ ; \alpha^N)$ and $S(\rho , \alpha^N)$ are given by
$$
I_E (\rho \ ; \alpha^N , \beta_{\Lambda}^N) := S(\Phi_E (\rho \ ; \alpha^N \cup \beta_{\Lambda}^N) , \Phi_E (\rho \ ; \alpha^N)\otimes \Phi_E (\rho \ ; \beta_{\Lambda}^N))
$$
$$
I(\rho \ ; \alpha^N , \beta_{\Lambda}^N) := \sup \{ I_E (\rho \ ; \alpha^N , \beta_{\Lambda}^N) \ ; E = \{ E_n \} \}
$$
where the supremum of $I_E (\rho ; \alpha^N , \beta_{\Lambda}^N)$ is taken over possible choices $E = {E_n}$ of the Schatten decompositions of $\rho_i$.
$$
S_E (\rho \ ; \alpha^N) := \sum_n \lambda_n S(\otimes_{i=0}^{N-1}E_{n_i} , \otimes_{i' = 0}^{N-1}\rho_{i'} )
$$
$$
S(\rho , \alpha^N) := \sup \{ S_E (\rho \ ; \alpha^N) ; E = \{ E_n \} \}.
$$
\end{definition}

\noindent Now we state the definitions of quantum dynamical entropies in density case.
\begin{definition}
Then the quantum dynamical entropy and the quantum dynamical mutual entropy are given by
\begin{equation}
\tilde{S}(\rho \ ; \theta_{\mathcal{A}} , \alpha) := \lim_{N \to \infty}\frac{1}{N} S (\rho \ ; \alpha^N),
\end{equation}
\begin{equation}
\tilde{I}(\rho \ ; \Lambda^* , \theta_{\mathcal{A}}, \theta_{\mathcal{B}} , \alpha , \beta_{\Lambda}) := \lim_{N \to \infty} \frac{1}{N} I(\rho \ ; \alpha^N , \beta_{\Lambda}^N).
\end{equation}
\end{definition}
$\ \ $There have been several attempts at defining dynamical mutual entropy on operator algebras. In \cite{muto}, Muto and Watanabe introduced the quantum dynamical mutual entropy whose time evolutions are given by c.p. maps. Furthermore, quantum Markovian dynamical mutual entropy was formulated by Ohmura and Watanabe on von Neumann algebras \cite{ow}.


\section{Comparison of Modulated States}
Optical communication using photons (laser beam) as carrier waves is currently widely used. In optical communication, one have to properly modulate the signal to the optical device. \\
In this section, we discuss the efficiency of two optical modulations (PPM, PWM) with the quantum states by using the entropy ratio given by the quantum dynamical entropy and the quantum dynamical mutual entropy. \\

\noindent Let $\{\alpha_1 , \cdots , \alpha_M \}$ be an alphabet set constructing the input signals and $\{ E_1 , \cdots , E_N \}$ be the set of all one dimensional projections on a space $\mathcal{H}$ satisfying $E_n \perp E_m \ (n \neq m)$. Then $E_n$ corresponds to the alphabet $a_n$.\\
$\mathfrak{S}_0$ denotes the set of all density operators on $\mathcal{H}$:
$$
\mathfrak{S}_0 := \{ \rho_0 = \sum_{n=1}^N \lambda_n E_n \ ;\  \rho_0 \ge 0 \ ,\ \mathrm{Tr} \rho_0 = 1\},
$$
where $\rho_0$ represents a state of the quantum input system.  In order to send information effectively, $\rho_0$ is transmitted from the quantum input system $\mathfrak{S}_0$ to the quantum modulator.\\
\noindent Let $M$ be an modulator and $\{ E_1^{(M)}, \cdots , E_N^{(M)} \}$ be the set of one dimensional projections on a Hilbert space $\mathcal{H}_{M}$ for modulated signals satisfying $E_n^{(M)} \perp E_m^{(M)} (n \neq m)$, and we represent the set of all density operators on $\mathcal{H}_M$ by
$$
\mathfrak{S}_0^{(M)} := \{ \rho_0^{(M)} = \sum_{n=1}^N \mu_n E_n \ ;\ \rho_0^{(M)} \ge 0 \ ,\ \mathrm{Tr} \rho_0^{(M)} = 1\},
$$
where $\rho_0^{(M)}$ represents a modulated state of the quantum input system.There are several expressions of quantum modulations \cite{effi}. In this paper, we give the modulated states by means of the photon number states.\\

\noindent Let  $\gamma_M$ be a c.p.u. map from $\mathcal{A}_0$ to $\mathcal{A}$. Then we obtain the c.p. channel $\gamma_M^* (E_n) = E_n^{(M)}$. The map $\gamma_M^* : \mathfrak{S}_0 \to \mathfrak{S}_0^{(M)}$ represents a modulator. Moreover, if  $\gamma_{M}^* (E_n) = E_n^{(M)}$ is a modulator from $\mathfrak{S}_0$ to $\mathfrak{S}_0^{(M)}$ and $\gamma_{(M)}^* (E_n) \perp \gamma_{(M)}^* (E_m)$ holds for any orthogonal $E_n \in \mathfrak{S}_0$, $\gamma_{M}^* $ is called an {\it ideal modulator} $IM$ and denoted by $\gamma_{IM}^*$.\\
Several examples of ideal modulators are given as follows:
\begin{definition}
For any $E_n \in \mathfrak{S}_0$, the PAM (Pulse Amplitude Modulator) is defined by
\begin{equation}
\gamma_{PAM}^* (E_n) :=  E_n^{(PAM)} = |n\rangle \langle n|
\end{equation}
where $|n\rangle \langle n |$ is the $n$ photon number state on $\mathcal{H}$.
\end{definition}

\begin{definition}
For any $E_n \in \mathfrak{S}_0$, the PWM (Pulse Width Modulator) is defined by
$$
\gamma_{PWM}^* (E_n) :=  E_n^{(PWM)}
$$
\begin{equation}
= E_d^{(PAM)} \otimes \cdots \otimes E_d^{(PAM)} \otimes \overbrace{E_d^{(PAM)}}^{n-th} \otimes E_0^{PAM} \otimes \cdots \otimes E_0^{PAM}
\end{equation}
where $E_0^{(PAM)} = |0\>\<0|$ is a vacuum state and $ E_d^{(PAM)} = | d\rangle \langle d |$.
\end{definition}

\begin{definition}
For any $E_n \in \mathfrak{S}_0$, the PPM (Pulse Position Modulator) is defined by
$$
\gamma_{PPM}^* (E_n) :=  E_n^{(PPM)}
$$
\begin{equation}
= E_0^{(PAM)} \otimes \cdots \otimes E_0^{(PAM)} \otimes \overbrace{E_d^{(PAM)}}^{n-th} \otimes E_0^{PAM} \otimes \cdots \otimes E_0^{PAM}
\end{equation}
\end{definition}

\noindent Now we calculate the quantum dynamical entropies for the modulated states (PWM, PPM) expressed by the photon number states as above.\\

\noindent {\bf PWM}\\
The finite sequence of c.p.u. maps  $\alpha_{(PWM)}^N$ and $\beta_{(PWM)}^N$ are given by
$$
\alpha_{(PWM)}^N := (\alpha \circ \tilde{\gamma}_{(PWM)} , \theta_{\mathcal{A}} \circ \alpha \circ \tilde{\gamma}_{(PWM)}, \cdots , \theta_{\mathcal{A}}^{N-1} \circ \alpha \circ \tilde{\gamma}_{(PWM)}),
$$
$$
\beta_{(PWM)}^N := (\tilde{\gamma}_{(PWM)} \circ \tilde{\Lambda} \circ \beta , \tilde{\gamma}_{(PWM)}) \circ \tilde{\Lambda} \circ \theta_{\mathcal{B}} \circ \beta , \cdots , \tilde{\gamma}_{(PWM)} \circ \tilde{\Lambda} \circ \theta_{\mathcal{B}}^{N-1} \circ \beta ),
$$
where we put $\tilde{\Lambda} := \otimes_{i=-\infty}^{\infty} \Lambda$ and $\tilde{\gamma}_{PWM} := \otimes_{i=-\infty}^{\infty} \gamma_{(PWM)}$.\\
Now 
$$
\rho := \otimes_{i=-\infty}^{\infty} \rho_i \in \otimes_{i=-\infty}^{\infty} \mathfrak{S}_i
$$
denotes a stationary input state. Furthermore, let $\rho_i = \sum_{n_i = 1}^{M} \lambda_{n_i} E_{n_i}$ be the Schatten decomposition of $\rho_i$. And we define $E_{n_i}^{(PWM)}$ as follows
$$
E_{n_i}^{(PWM)} := \otimes_{j=1}^{M} E_{d \tau_{j, n_i}}^{(PAM)},
$$
\begin{equation*}
\tau_{j,n_i} := \begin{cases} 
	             1\qquad (j \le n_i) \\
	             0 \qquad (j > n_i)
	\end{cases}
\end{equation*}
Then the compound states of input and output are given by  
\begin{eqnarray*}
\Phi_E (\alpha_{(PWM)}^{N}) &=& \otimes_{i=0}^{N-1} \gamma_{(PWM)}^* \circ \alpha^* \circ \theta_{\mathcal{A}}^{*i} (\rho)\\
&=& \otimes_{i=0}^{N-1} \gamma_{(PWM)}^* (\rho_i),\\
\Phi_E (\beta_{\Lambda(PWM)}^N) &=& \otimes_{i=0}^{N-1} \circ \beta^* \circ \theta_{\mathcal{B}}^*i (\rho) \circ \lambda^* \circ \gamma_{(PWM)}^* (\rho) \\
&=& \otimes_{i=0}^{N-1} \Lambda^* \circ \gamma_{(PWM)}^* (\rho_i)
\end{eqnarray*}
respectively.\\
When $\Lambda^*$ is an attenuation channel (\ref{atte}), the compound states through the channel $\tilde{\Lambda}^*$ becomes
$$
\Phi_E (\alpha_{(PWM)}^N \cup \beta_{\Lambda (PWM)}^N)
$$
$$
= \sum_{n_0 = 1}^M \cdots \sum_{n_{N-1}=1}^M (\prod_{k=0}^{N-1} \lambda_{n_k}) (\otimes_{i=0}^{N-1} (\otimes_{j=1}^M E_{d \tau_{j , n_i}}^{(PAM)})) \otimes
 (\otimes_{i' = 0}^{N-1} (\otimes_{j' = 1}^M E_{d, \tau_{j' , n_{i'}}}^{(PAM)})),
$$
$$
\Phi_E (\alpha_{(PWM)}^N )\otimes \Phi_E(\beta_{\Lambda (PWM)}^N) 
$$
\begin{eqnarray*}
&=& \sum_{n_0 = 1}^M \cdots \sum_{n_{N-1}=1}^M (\prod_{k=0}^{N-1} \lambda_{n_k}) (\otimes_{i=0}^{N-1}(\otimes_{j=1}^M E_{d \tau_{j, n_i}}^{(PAM)})) \otimes \\
& & \sum_{n'_0 = 1}^M \cdots \sum_{n'_{N-1}=1}^M (\prod_{k'=0}^{N-1} \lambda_{n'_{k'}})(\otimes_{j'=0}^{N-1}(\otimes_{j'=1}^M \Lambda^* E_{d\tau_{j', n'_{i'}}}^{(PAM)})).
\end{eqnarray*}
After the calculation we get
\begin{equation}
\tilde{S} (\rho \ ;\ \theta_{\mathcal{A}},\ \alpha_{(PWM)} ) = - \sum_{n=1}^M \lambda_n \log \lambda_n ,
\end{equation}
\begin{equation}
\tilde{I} (\rho \ ; \Lambda^* , \theta_{\mathcal{A}}, \theta_{\mathcal{B}}, \alpha_{(PWM)}, \beta_{\Lambda(PWM)}) = - \sum_{n=1}^M(1-(1-\eta)^d)^n \lambda_n \log \lambda_n.
\end{equation}

\noindent {\bf PPM}\\
Under the same conditions, similarly we obtain
\begin{equation}
\tilde{S} (\rho \ ;\ \theta_{\mathcal{A}},\ \alpha_{(PPM)} ) = - \sum_{n=1}^M \lambda_n \log \lambda_n ,
\end{equation}
\begin{equation}
\tilde{I} (\rho \ ; \Lambda^* , \theta_{\mathcal{A}}, \theta_{\mathcal{B}}, \alpha_{(PPM)}, \beta_{\Lambda(PPM)}) = - (1-(1-\eta)^d) \sum_{n=1}^M \lambda_n \log \lambda_n 
\end{equation}
(e.g. see \cite{acnote}).\\

\noindent Since $0 \le \eta \le 1$ (\ref{eta}), one can see the following result.

\begin{theorem}
Under the above settings, there holds
\begin{equation}
\tilde{I} (\rho\ ; \Lambda^* , \theta_{\mathcal{A}}, \theta_{\mathcal{B}}, \alpha_{(PPM)}, \beta_{\Lambda(PPM)}) \geq \tilde{I} (\rho \ ; \Lambda^* , \theta_{\mathcal{A}}, \theta_{\mathcal{B}}, \alpha_{(PWM)}, \beta_{\Lambda(PWM)}).
\end{equation}
\end{theorem}

\noindent Furthermore, an important measure to consider the transmission efficiency of modulators is the {\it entropy ratio} \cite{itsuse} \cite{effi}. For the above entropies and an ideal modulators $IM$, the entropy ratio is given by
\begin{equation}
r (\rho \ ; \Lambda^* , \theta_{\mathcal{A}}, \theta_{\mathcal{B}}, \alpha_{(IM)} \beta_{\Lambda(IM)}) := \frac{\tilde{I} (\rho \ ; \Lambda^* , \theta_{\mathcal{A}}, \theta_{\mathcal{B}}, \alpha_{(IM)}, \beta_{\Lambda(IM)})}{\tilde{S} (\rho \ ;\ \theta_{\mathcal{A}},\ \alpha_{(IM)})}.
\end{equation}
From the fundamental inequalities (\ref{fund2}),
\begin{equation}
0 \le r (\rho \ ; \Lambda^* , \theta_{\mathcal{A}}, \theta_{\mathcal{B}}, \alpha_{(IM)} \beta_{\Lambda(IM)}) \le 1.
\end{equation}
Therefore, the entropy ratio is a measure which gives the rate of the amount of information correctly transmitted from the input to the output system. Thus, by fixing $\Lambda^* , \theta_{\mathcal{A}},$ and $\theta_{\mathcal{B}}$, we can compare the transmission efficiency of the modulators.\\ 

\noindent Now we state the main result in this paper. 
\begin{theorem}
For an initial state $\rho$, the following inequality holds:
\begin{equation}
r (\rho \ ; \Lambda^* , \theta_{\mathcal{A}}, \theta_{\mathcal{B}}, \alpha_{(PPM)} \beta_{\Lambda(PPM)}) \geq r (\rho \ ; \Lambda^* , \theta_{\mathcal{A}}, \theta_{\mathcal{B}}, \alpha_{(PWM)} \beta_{\Lambda(PWM)}).
\end{equation}
\end{theorem}
\begin{proof}
According to the equality
$$
\tilde{S} (\rho \ ;\ \theta_{\mathcal{A}},\ \alpha_{(PPM)}) = \tilde{S} (\rho \ ;\ \theta_{\mathcal{A}},\ \alpha_{(PWM)} )  
$$
and Theorem 5, we obtain the above inequality.
\end{proof}
This result tells us that, under the above conditions, the loss of the average amount of information is smaller in the case of modulating the input quantum state using PPM than in the case of PWM.


\end{document}